\documentclass[a4paper,10pt]{cypart}
\usepackage{amsthm,mathtools}

\def\boldA{\boldsymbol A}
\def\boldb{\boldsymbol b}
\def\boldX{\boldsymbol X}
\def\boldY{\boldsymbol Y}
\def\boldZ{\boldsymbol Z}
\def\boldmu{\boldsymbol\mu}
\def\boldSigma{\boldsymbol\Sigma}

\newtheorem{thm}{Theorem}
\newtheorem{lem}[thm]{Lemma}
\newtheorem{prop}[thm]{Proposition}

\title{\Large\bf A Constructive Algebraic Proof of Student's Theorem}
\author{Yiping Cheng}
\affil{\vskip -3pt School of Electronic and Information Engineering\\
Beijing Jiaotong University, Beijing 100044, China\\
\texttt{ypcheng@bjtu.edu.cn}}

\begin{abstract}
Student's theorem is an important result in statistics which
states that for normal population, the sample variance is
independent from the sample mean and has a chi-square
distribution. The existing proofs of this theorem either overly
rely on advanced tools such as moment generating functions, or
fail to explicitly construct an orthogonal matrix used in the
proof. This paper provides an elegant explicit construction of
that matrix, making the algebraic proof complete. The constructive
algebraic proof proposed here is thus very suitable for being
included in textbooks.
\end{abstract}
\keywords{sample variance; chi-square distribution;
t-distribution; statistical education}

\begin{document}

\maketitle

\section{Student's Theorem}

In mathematical statistics, there is a well-known theorem about
the sample variance of a random sample from a normal distribution.
This theorem is directly related to the discovery of the
t-distribution by statistician William Sealy Gosset (1876-1937),
known as ``Student", a pseudonym he used when he published his
paper. Therefore, this theorem is often referred to as Student's
theorem. Let $N(\mu,\sigma^2)$ denote the normal distribution with
mean $\mu$ and variance $\sigma^2$. Then the theorem reads as
follows.

\begin{thm}[Student's Theorem]
Let $X_1,\ldots,X_n$ be a random sample from the distribution
$N(\mu,\sigma^2)$, i.e., they all have that distribution and are
mutually independent. Define the random variables \begin{eqnarray}
\overline{X} &= & {\sum_{i=1}^n X_i\over n}, \\
S^2 & = & {\sum_{i=1}^n (X_i-\overline{X})^2\over n-1}.
\end{eqnarray} Then

1. $\overline{X}$ has distribution $N(\mu,{\sigma^2\over n})$.

2. $\overline{X}$ and $S^2$ are independent.

3. ${(n-1)S^2\over\sigma^2}$ has distribution $\chi^2(n-1)$.
\end{thm}

This theorem is equivalent to the following version where the
general normal distribution is replaced by standard normal
distribution.

\begin{thm}[Student's Theorem, Standardized Version]
Let $Z_1,\ldots,Z_n$ all have distribution $N(0,1)$ and are
mutually independent. Define the random variables \begin{eqnarray}
\overline{Z} &= & {\sum_{i=1}^n Z_i\over n}, \\
W & = & \sum_{i=1}^n (Z_i-\overline{Z})^2.
\end{eqnarray} Then

1. $\sqrt{n}\thinspace\overline{Z}$ has distribution $N(0,1)$.

2. $\overline{Z}$ and $W$ are independent.

3. $W$ has distribution $\chi^2(n-1)$.
\end{thm}

Since these two versions are equivalent, and it is easier to
formulate a proof of the standardized version, in the rest of the
paper the standardized version will be used when we give our
proof.

\section{Literature Proofs of Student's Theorem}

To the author's best knowledge, the original paper of Gosset is
not currently available to the general public, so we do not know
if it contained a proof of the above theorem. However, it is
believed that even if such a ``proof" did exist, it could hardly
be regarded as a proof by today's standard, because the
mathematically rigorous theory of probability only began to emerge
in 1930s. We therefore should look into the modern literature,
mainly textbooks, for proofs of Student's theorem. In one way or
another, all the proofs rely on two important theorems of
multivariate normal distribution, whose proofs require a very deep
mathematical tool: moment-generating functions (m.g.f. in the
sequel), or alternatively, characteristic functions. These two
theorems are familiar to the majority of statistics students. They
are given here as lemmas.

\begin{lem}
Let random variables $X_1,\ldots,X_n$ have the multivariate normal
distribution with mean $\boldmu$ and covariance matrix
$\boldSigma$. Let $\boldY=[Y_1,\ldots,Y_m]^T=\boldA\boldX+\boldb$,
where $\boldA$ is an $m\times n$ {\em full row-rank\/} constant
matrix, $\boldX=[X_1,\ldots,X_n]^T$, and
$\boldb=[b_1,\ldots,b_m]^T$ is a constant column vector. Then
$Y_1,\ldots,Y_m$ have the multivariate normal distribution with
mean $\boldA\boldmu+\boldb$ and covariance matrix
$\boldA\boldSigma\boldA^T$.
\end{lem}

\begin{lem}
Let random variables $X_1,\ldots,X_n$ have the multivariate normal
distribution with mean $\boldmu$ and covariance matrix
$\boldSigma$. Define random vectors $\boldX$, $\boldX_1$, and
$\boldX_2$ as
\[\boldX^T=[\underbrace{X_1,\ldots,X_r}_{\boldX_1^T},\underbrace{X_{r+1},\ldots,X_n}_{\boldX_2^T}].\]
Partition $\boldSigma$ as \[\boldSigma=\left[\begin{array}{cc}
\underbracket{\boldSigma_{11}}_{r\times r} & \underbracket{\boldSigma_{12}}_{r\times(n-r)} \\
\underbracket{\boldSigma_{12}^T}_{(n-r)\times r} &
\underbracket{\boldSigma_{22}}_{(n-r)\times(n-r)}
\end{array}\right].\]
Then $\boldX_1$ and $\boldX_2$ are independent if and only if
$\boldSigma_{12}=\mathbf{0}$.
\end{lem}

A consequence of Lemma 4 is the following proposition, which we
will use later.

\begin{prop}
Let random variables $X_1,\ldots,X_n$ have multivariate normal
distribution with covariance matrix $\boldSigma$. Then
$X_1,\ldots,X_n$ are mutually independent if and only if
$\boldSigma$ is a diagonal matrix.
\end{prop}

After looking into a number of renowned modern statistical
textbooks, which are supposed to have incorporated the latest
developments in the whole literature on this subject, we found two
typical proofs. They are commented below.

\paragraph{Proof in \cite[Section 3.6.3]{hogg13}.}
This proof first shows the independence of $\overline{X}$ and
$S^2$ using Lemma 4, then it shows that ${(n-1)S^2\over\sigma^2}$
has distribution $\chi^2(n-1)$ using an argument that invokes
m.g.f. a further time. This, we believe, is a drawback because the
typical reader, who is usually only a sophomore, is not expected
to have the skill of directly dealing with m.g.f.. There is a
similar proof in \cite[Section 8.5]{walpole12}, which we consider
is somewhat less rigorous than the one in \cite{hogg13}.

\paragraph{Proof in \cite[Section 7.3]{groot89} and \cite[Section
8.3]{groot12}.} This proof shows the independence and the
$\chi^2(n-1)$ distribution in one single step. It defines a new
vector $\boldY=\boldsymbol O\boldsymbol Z$, where $\boldsymbol O$
is an orthogonal matrix and the first row of $\boldsymbol O$ is
$[{1\over \sqrt{n}},\ldots,{1\over \sqrt{n}}]$, so that
$Y_1=\sqrt{n}\thinspace\overline{Z}$ and the sum of squares of the
other entries of $\boldY$ is $W$. This proof is algebraic, without
using advanced tools, and hence is much simpler and easier to
understand than the proof in \cite{hogg13}. However, there is
still a little drawback of this proof: it is nonconstructive in
that it only states the existence of the orthogonal matrix
$\boldsymbol O$, without giving it specifically. While not
affecting the rigor of the proof, this drawback does hurt its
pedagogical value.

\section{Proposed Constructive Algebraic Proof}

We consider the proof in \cite{groot89,groot12} nearly perfect,
and we seek to make it fully perfect by fixing its drawback we
just mentioned, i.e. by explicitly constructing the $\boldsymbol
O$ matrix. In fact, in \cite[page 478]{groot12} the Gram-Schmidt
orthogonalization method is suggested for constructing the
$\boldsymbol O$ matrix, but no hint is given about the choice of
the starting matrix. We tried that method with the starting matrix
being the matrix obtained by replacing the first row of the
identity matrix by $[{1\over \sqrt{n}},\ldots,{1\over \sqrt{n}}]$,
and we found that the resulting orthogonal matrix is very ugly and
prohibitively difficult to describe. Therefore we tend to believe
that Gram-Schmidt orthogonalization is not an elegant method of
construction for our purpose here. However, we finally succeeded
in finding an elegant construction. Let us now illustrate it by a
few base examples.
\begin{equation}
\boldsymbol O_2 = \left[\begin{array}{cc}
{1\over\sqrt{2}} & {-1\over\sqrt{2}} \\
{1\over\sqrt{2}} & {1\over\sqrt{2}}
\end{array} \right].
\end{equation}
\begin{equation}
\boldsymbol O_3 = \left[\begin{array}{ccc}
{1\over\sqrt{2}} & {-1\over\sqrt{2}} & 0\\
{1\over\sqrt{6}} & {1\over\sqrt{6}} & {-2\over\sqrt{6}} \\
{1\over\sqrt{3}} & {1\over\sqrt{3}} & {1\over\sqrt{3}}
\end{array} \right].
\end{equation}
\begin{equation}
\boldsymbol O_4 = \left[\begin{array}{cccc}
{1\over\sqrt{2}} & {-1\over\sqrt{2}} & 0 & 0 \\
{1\over\sqrt{6}} & {1\over\sqrt{6}} & {-2\over\sqrt{6}} & 0\\
{1\over\sqrt{12}} & {1\over\sqrt{12}} & {1\over\sqrt{12}} & {-3\over\sqrt{12}}\\
{1\over\sqrt{4}} & {1\over\sqrt{4}} & {1\over\sqrt{4}} &
{1\over\sqrt{4}}
\end{array} \right].
\end{equation}
\begin{equation}
\boldsymbol O_5 = \left[\begin{array}{ccccc}
{1\over\sqrt{2}} & {-1\over\sqrt{2}} & 0 & 0 & 0\\
{1\over\sqrt{6}} & {1\over\sqrt{6}} & {-2\over\sqrt{6}} & 0 & 0\\
{1\over\sqrt{12}} & {1\over\sqrt{12}} & {1\over\sqrt{12}} & {-3\over\sqrt{12}} & 0\\
{1\over\sqrt{20}} & {1\over\sqrt{20}} & {1\over\sqrt{20}} & {1\over\sqrt{20}} & {-4\over\sqrt{20}}\\
{1\over\sqrt{5}} & {1\over\sqrt{5}} & {1\over\sqrt{5}} &
{1\over\sqrt{5}} & {1\over\sqrt{5}}
\end{array} \right].
\end{equation}

\vskip 8pt The general method of construction is contained in the
following key lemma.
\begin{lem} For every integer $n\ge 2$, define the matrix
$\boldsymbol O_n=[o_{ij}]_{n\times n}$ by: \vskip 3pt

\begin{equation}\mbox{For each }i\mbox{ with }1\le i\le n-1,\mbox{
 }o_{ij}=\left\{\begin{array}{cl}
{1\over\sqrt{i(i+1)}}, & \mbox{ for }j\le i,  \\
{-i\over\sqrt{i(i+1)}}, & \mbox{ for }j=i+1, \\
0, & \mbox{otherwise;}
\end{array}\right.\end{equation}

\begin{equation} o_{nj}={1\over \sqrt{n}}\mbox{ for all }1\le
j\le n.\end{equation}

\vskip 3pt Then $\boldsymbol O_n$ is orthogonal, i.e. $\boldsymbol
O_n \boldsymbol O_n^T=\boldsymbol I$.
\end{lem}
\begin{proof}
Let $\boldsymbol P=\boldsymbol O_n \boldsymbol
O_n^T=[p_{ij}]_{n\times n}$.

1) If $1\le i \le n-1$, then $p_{ii}=\sum_{k=1}^n
o_{ik}^2=\sum_{k=1}^{i} {1\over i(i+1)}\thinspace + \thinspace
{i^2\over i(i+1)}=1$.

2) $p_{nn}=\sum_{k=1}^n o_{nk}^2=\sum_{k=1}^n {1\over n}=1$.

3) If $1\le i \le n-1$, then $p_{in}=p_{ni}=\sum_{k=1}^n
o_{ik}o_{nk}=\sum_{k=1}^{i} {1\over \sqrt{i(i+1)}}{1\over
\sqrt{n}}\thinspace + \thinspace {-i\over \sqrt{i(i+1)}}{1\over
\sqrt{n}} = 0$.

4) If $1\le i\neq j\le n-1$, without loss of generality, let us
assume $i<j$, then \[p_{ij}=p_{ji}=\sum_{k=1}^n o_{ik}o_{jk} =
\sum_{k=1}^{i+1}
o_{ik}o_{jk}={1\over\sqrt{j(j+1)}}\sum_{k=1}^{i+1} o_{ik}=0.\]

Thus we have shown $\boldsymbol O_n \boldsymbol O_n^T=\boldsymbol
I$.\end{proof} \vskip 4pt

Now for the sake of self-completeness of this paper, we here give
a proof of Theorem 2. It uses the same idea as the proof in
\cite[page 478]{groot12} except for our explicit construction of
$\boldsymbol O$ and a few minor details.

\begin{proof}[Proof of Theorem 2]
Denote $\boldZ=[Z_1,\ldots,Z_n]^T$. Define the random vector
$\boldY=[Y_1,\ldots,Y_n]^T$ by
\begin{equation}
\boldY =\boldsymbol O_n \boldZ
\end{equation} where $\boldsymbol O_n$ is defined by (9,10).

From (10) we know that \begin{equation} Y_n=\sum_{i=1}^n
{Z_i\over\sqrt{n}}=\sqrt{n}\thinspace\overline{Z}.\end{equation}
It is obvious that $Y_n$ has the distribution $N(0,1)$.

Furthermore, by Lemma 6, we have \[\sum_{i=1}^n Y_i^2 =
\boldY^T\boldY = \boldZ^T \boldsymbol O_n^T \boldsymbol O_n
\boldZ=\boldZ^T \boldZ =\sum_{i=1}^n Z_i^2.\] Therefore,
\[\sum_{i=1}^{n-1} Y_i^2=\sum_{i=1}^{n} Y_i^2\thinspace-\thinspace Y_n^2=\sum_{i=1}^{n} Z_i^2\thinspace-\thinspace n\overline{Z}^2
=\sum_{i=1}^n (Z_i-\overline{Z})^2.\] We have thus obtained the
relation \begin{equation}W=\sum_{i=1}^n
(Z_i-\overline{Z})^2=\sum_{i=1}^{n-1} Y_i^2.\end{equation}

By Lemma 3, $Y_1,\ldots,Y_n$ have multivariate normal distribution
with covariance matrix $\boldsymbol O_n I \boldsymbol O_n^T=I$,
which is diagonal, therefore by Proposition 5,
\begin{equation}
Y_1,\ldots,Y_n\mbox{ all have the }N(0,1)\mbox{ distribution and
are mutually independent}.
\end{equation} Since $W$ is entirely based on
$Y_1,\ldots,Y_{n-1}$, and $\overline{Z}={Y_n\over\sqrt{n}}$, (14)
implies that $W$ and $\overline{Z}$ are independent.

Finally, (13) and (14) together imply that $W$  has distribution
$\chi^2(n-1)$.
\end{proof}

\section{Conclusion}

The proof proposed here of Student's theorem is algebraic and
fully constructive. To our best knowledge, such a construction has
not appeared in the literature before. A constructive proof is
expected to make the reader more comfortable and consequently
enhance their understanding of this important result. We believe
this paper to be of significant pedagogical value in statistical
education, and hope the construction proposed here to be included
in future textbooks.

\bibliographystyle{unsrt}

\begin{thebibliography}{1}

\bibitem{hogg13}
R.V. Hogg, J.W. McKean, and A.T. Craig.
\newblock {\em Introduction to Mathematical Statistics}.
\newblock Pearson, 7th edition, 2013.

\bibitem{walpole12}
R.E. Walpole, H.~Myers, S.L. Myers, and K.~Ye.
\newblock {\em Probability and Statistics for Engineers and Scientists}.
\newblock Prentice Hall, 9th edition, 2012.

\bibitem{groot89}
M.H. DeGroot.
\newblock {\em Probability and Statistics}.
\newblock Addison-Wesley, 2nd edition, 1989.

\bibitem{groot12}
M.H. DeGroot and M.J. Schervish.
\newblock {\em Probability and Statistics}.
\newblock Addison-Wesley, 4th edition, 2012.

\end{thebibliography}

\end{document}